\documentclass[11pt]{article}

\usepackage{amsmath,amsthm,amssymb,amsfonts,verbatim}
\usepackage{fullpage}

\usepackage[dvipsnames,usenames]{color}

 \providecommand{\F}{\mathbb{F}}

\newtheorem{theorem}{Theorem}[section]
\newtheorem{cor}[theorem]{Corollary}
\newtheorem{prop}[theorem]{Proposition}
\newtheorem{lem}[theorem]{Lemma}

\newtheorem{ex}[theorem]{Example}

\theoremstyle{remark}
\newtheorem{rmk}{Remark}

\renewcommand{\epsilon}{\varepsilon}
\renewcommand{\le}{\leqslant}
\renewcommand{\ge}{\geqslant}

\def\proj{\mathrm{proj}}

\def\Im{\mathrm{Im}}

\def\ZZ{\mathbb{Z}}

\def\PP{\mathbb{P}}

\def \mL {\mathcal{L}}

\def \mP {\mathcal{P}}

\def \mS {\mathcal{S}}

\def\Pin{{P_{\infty}}}
\def\Tr{{\rm Tr}}

\def\Span{{\rm Span}}

\newcommand{\Ga}{\alpha}
\newcommand{\Gb}{\beta}

\newcommand{\Gg}{\gamma}     
\newcommand{\Gd}{\delta}     
\newcommand{\Ge}{\epsilon}

\newcommand{\Gl}{\lambda}    
\newcommand{\Go}{\omega}     \newcommand{\GO}{\Omega}
     
\newcommand{\Gt}{\tau}

\def \bc {{\bf c}}

\def \fb {{\mathfrak{b}}}

\def\supp {{\rm supp }}
\def \res {{\rm res}}

\def\g{{\mathfrak{g}}}
\def\proj{{\rm proj}}

\begin{document}

\title{{\bf Efficiently repairing algebraic geometry  codes }}
\author{Lingfei Jin\thanks{Shanghai Key Laboratory of Intelligent Information Processing, School of Computer Science, Fudan University, Shanghai 200433, China. {\tt lfjin@fudan.edu.cn}.}
\and Yuan Luo \thanks{Department of Computer Sciences and Engineering, Shanghai Jiaotong University, Shanghai 200240, P. R. China. {\tt luoyuan@cs.sjtu.edu.cn}.}
\and Chaoping Xing\thanks{Division of Mathematical Sciences, School of Physical \&  Mathematical Sciences, Nanyang Technological University, Singapore. {\tt xingcp@ntu.edu.sg}}}
\date{}

\maketitle

\begin{abstract}
Minimum storage regenerating  codes have minimum storage of data in each node and therefore are maximal distance separable (MDS for short) codes. Thus, the number of nodes is upper bounded by $2^{\fb}$, where $\fb$ is the bits of data stored in each node. From both theoretical and practical points of view (see the details in Section 1), it is natural to consider regenerating codes that nearly have  minimum storage of data, and meanwhile the number of nodes is unbounded. One of the candidates for such regenerating  codes  is an algebraic geometry code.
In this paper, we generalize the repairing algorithm of Reed-Solomon codes given in \cite[STOC2016]{GW16} to  algebraic geometry codes and present an efficient repairing algorithm for arbitrary one-point algebraic geometry codes. By applying our repairing algorithm to the one-point algebraic geometry codes based on the Garcia-Stichtenoth tower, one can  repair a code of rate $1-\Ge$ and length $n$ over $\F_{q}$ with bandwidth $(n-1)(1-\Gt)\log q$ for any $\Ge=2^{(\Gt-1/2)\log q}$ with a real $\tau\in(0,1/2)$. In addition,  storage in each node for an algebraic geometry code is close to the minimum storage.  Due to nice structures of Hermitian curves, repairing of Hermitian codes is also investigated. As a result, we are able to show that algebraic geometry codes are regenerating codes with good parameters. An example  reveals that  Hermitian codes outperform  Reed-Solomon codes for certain parameters.\end{abstract}

\section{Introduction}
In a distributed storage system, a large file is encoded and distributed over many nodes. When
a few nodes fail, one should  be able to rebuild  replacement nodes efficiently by using information from the
remaining active nodes. The problem of recovering the failed nodes exactly is known as the exact repair problem.
The exact repair problem and regenerating codes were first introduced in \cite{DGWWR10}. It was shown in \cite{DGWWR10} that there is a trade-off between storage and repair bandwidth. Codes lying on this tradeoff are called
regenerating codes. There are two special cases of regenerating codes that are  interesting from the theoretical point of view. One is called minimum bandwidth regenerating (MBR for short) code  in which
the minimum repair bandwidth is needed to repair the failed nodes. The other case is called minimum storage regenerating (MSR for short) code that corresponds
to the minimum storage. We refer \cite{DRWC11} for an excellent survey.

\subsection{Definition}
Let us give a formal definition of regenerating codes. In this paper, $q$ is a prime power and $\F_q$ denotes  the finite field $q$ elements.

A subspace $C$ of $\F_q^n$ is a called a {\it strong $[n,m,d]$-regenerating code} with the secondary parameters $(\mathfrak{b}, M,B)$ if
\begin{itemize}
\item[(i)] $\mathfrak{b}=\log q$ and $M$  is the total amount of data stored, i.e., $M=\log|C|$, where the logarithm is of base $2$;
\item[(ii)] given any codeword $\bc=(c_1,c_2,\dots,c_n)$  of $C$ and a coordinate $c_j\in\F_q$, any set $\{c_i\}_{i\in I}$  of coordinates with $I\subseteq [n]\setminus\{j\}$ and $|I|=m$ can repair $c_j$, where $[n]$ stands for the set $\{1,2,\dots,n\}$;
\item[(iii)] given any codeword $\bc=(c_1,c_2,\dots,c_n)$  of $C$ and a coordinate $c_j\in\F_q$, any set $\{c_i\}_{i\in I}$  of coordinates with $I\subseteq [n]\setminus\{j\}$ and $|I|=d$, from each of which only $B/d$ bits are  downloaded, can recover $c_j$.
\end{itemize}
The above definition requires downloading $B/d$ bits equally from each of $d$ coordinates. If we replace this condition by the condition that the total downloaded bits from the $d$ coordinates is bandwidth $B$, then the code $C$ is called a {\it weak $[n,m,d]$-regenerating code} with the secondary parameters $(\fb, M,B)$.

\subsection{Motivation}
In literature, most of researchers focused on either MSR or MBR codes. As MSR codes are MDS codes, the number of nodes is upper bounded by  $2^{\fb}$, where $\fb$ is the bits of data stored in each node.  Although the number $2^{\fb}$ of nodes is already big enough in current network, it is of theoretical interest to  study regenerating codes with unbounded number of nodes and near minimum storage.

From practical point of view, a linear secret sharing scheme (LSSS for short) can be viewed as an $[n,m,d]$-regenerating  code in which a secret is  the data in the erased node. The current LSSS requires to recover this missing data from any $m$ active nodes by downloading the whole data in each of the $m$ nodes. By applying an $[n,m,d]$-regenerating  code instead of classical linear code, we are able to recover the secret by obtaining part of data from any $d$ nodes. Thus, a regenerating  code reduces the total downloading bandwidth for an LSSS. For many applications of LSSS,  the number $n-1$ of players is much bigger than $2^{\fb}$ (see \cite{CCCX09}). Actually, for some of these applications we fix the alphabet size $q$ and let the number of players tend to $\infty$. In this circumstance, MSR codes are not suitable. By applying algebraic geometry codes, the number of players is unbounded, while storage of data in each node is nearly minimum. More  precisely speaking, for an algebraic geometry code, the data stored in each node is $\fb$, while an MSR code with the same length $n$ and rate $R$ requires storage approximately equal to $\frac{R}{R+2^{-\fb}}\fb$.

\subsection{Known results and parameter regime}

Due to wide applications of Reed-Solomon (RS for short) codes, it is natural to consider RS codes for MSR codes.  However, as noted in \cite{GW16}, the traditional approach
with Reed-Solomon codes is not a good idea for the exact repair problem. In fact, we know several (non-RS) MDS codes which outperform the traditional Reed-Solomon approach \cite{RSK11,LTT16}.

On the other direction, Rashmi, Shah and  Kuma \cite{RSK11} made use of matrix-product to construct a class of (non-RS) MSR codes for rate up to $1/2$. After this work, people have been working on constructions of MSR codes with rate bigger than $1/2$ \cite{LTT16,CHL11,GTC14,TWB13}.

On the other hand, due to wide range of applications, it is still interesting to study some important existing codes such as Reed-Solomon codes and see what these codes can achieve in the scenario of the exact repairing. Guruswami and Wootters \cite{GW16} first gave a clever local repairing algorithm for Reed-Solomon codes. They showed that one requires less than $\log |C|$ bits in order to exactly repair a failed node. For instance, it shows in \cite{GW16} that the bandwidth is about $(n-1)/r$ if the rate is at most $1-1/r$. In particular, for a high rate MDS code a tight bound $O(n)$ on bandwidth is obtained, where $n$ is the length of the code.

For an $[n,m,d]$-MSR code with the secondary parameters $(\fb, M,B)$, we have $B=m\Ga$ and the  bandwidth must obey
\begin{equation}\label{eq:1.1}B\ge \frac{d\fb}{d-m+1}.\end{equation}
Determine lower bound on repairing bandwidth is a fundamental problem for regenerating codes \cite{DGWWR10,WDR07}.
The bandwidth given in \eqref{eq:1.1}   holds for functional repair as well and is  only possible when $\fb$ is sufficiently large  compared with $d-m$. More precisely  speaking, it was shown in  \cite{SR10,CJM+13} that the bound \eqref{eq:1.1} can be achieved when $\log q$ is exponential in $n-m$. Moreover, it was proved in \cite{GTC14} that the bound \eqref{eq:1.1} is not achievable if $\log q$ is less than $m^2$.  

In this paper, we consider  $[n,m,d]$-regenerating codes with the secondary parameters $(\fb, M,B)$ (not necessarily MDS codes) and the regime where $\fb$ is much smaller than $d-m$. In fact, in our setting, $\log q=\fb$ is a constant and $d-m$ tends to $\infty$. Thus, the bound \eqref{eq:1.1} is a constant and  not achievable by the following naive bound \eqref{eq:1.2}.

\subsection{LSSS}
We refer to \cite{CCCX09} for details in this section.

Let $C$ be a $q$-ary linear code in $\F_q^n$. For a codeword $\bc=(c_1,\dots, c_n)$ and $i\in[n]$, let the $i$th coordinate $c_i$  be the secret and let the rest of $n-1$ coordinates be the shares. $C$ is said to have $r$-reconstruction if, for any $I\subseteq[n]\setminus\{i\}$ with $|I|=r$, the projection $\proj_I(\bc)$ of $\bc$ at $I$ uniquely determines $\bc$ with probability $1$. $C$ is said to have $s$-privacy if, for any $I\subseteq[n]\setminus\{i\}$ with $|I|= s$,  the projection  $\proj_{I\cup\{i\}}(\bc)$ is uniformly distributed in $\F_q^{s+1}$. It was shown in \cite[Theorem 1]{CCCX09} that a linear code $C$ with distance $\delta$ and dual distance $\Gd^{\perp}$ has $(n-\Gd+1)$-reconstruction and $(\Gd^{\perp}-1)$-privacy. This implies that the bandwidth $B$ to recover $c_i$ must satisfy
\begin{equation}\label{eq:1.2}B\ge \Gd^{\perp}-2+\fb.\end{equation}
To prove the naive bound \eqref{eq:1.2}, assume that only a single bit from any $\Gd^{\perp}-2$ nodes are downloaded. As $C$ has $(\Gd^{\perp}-1)$-privacy, the remaining coordinate is free and hence $\fb=\log q$ bits is required.  In the case where $\fb$ is a constant and the number $n$ of nodes tends to $\infty$, the bound \eqref{eq:1.1} is a constant and the bound \eqref{eq:1.2} is proportional to $n$. This implies that
the bound \eqref{eq:1.2} is much better than the bound \eqref{eq:1.1} in this circumstance.

We make use of privacy to deduce the bound \eqref{eq:1.2}. However, in order to repair the failed node, i.e., reconstruct the failure node, we should consider reconstruction instead.  This means that one should download at least one bit from each of $(n-\Gd+1)$ nodes and hence the repairing bandwidth satisfies
\begin{equation}\label{eq:1.3}B\ge n-\Gd+1.\end{equation}
For MDS codes, we have $n-\Gd+1=\Gd^{\perp}-1$, i.e, the gap between the bounds \eqref{eq:1.2} and \eqref{eq:1.3} is $\log q-1$. This is a constant for a fixed $q$ and negligible when $n$ is large. However, for fixed $q$ and sufficiently large $n$, the gap between $n-\Gd+1$ and $\Gd^{\perp}-1$ is proportional to $n$. This means that in this case the bound \eqref{eq:1.2} is even worse than the bound \eqref{eq:1.3}.

\subsection{Our result }\label{subsec:com}
MSR codes achieve minimum storage with the best possible bandwidth, while MBR codes achieve minimum bandwidth with the best possible storage. A natural approach is to look into something between.

 A good class of candidates is algebraic geometry (AG for short) codes. First of all, algebraic geometry codes have both good privacy and reconstruction and their alphabet sizes can be very small. Secondly, these codes almost achieve minimum storage. Thirdly, algebraic geometry codes are a natural generalization of Reed-Solomon codes. The main purpose of this paper is to investigate efficient repairing of algebraic geometry codes.

We first provide an efficient repairing algorithm for arbitrary algebraic geometry codes and it shows that an algebraic geometry code is a good  regenerating code.  Applying this algorithm to the algebraic geometry codes based on the Garcia-Stichtenoth tower, we obtain the following result by Corollary \ref{cor:3.5}.
 \begin{quote}
 {\bf Main result:} {\it
 Let $q$ be a prime power and let $\Ge=2^{(\tau-1/2)\fb}\ge \frac2{\sqrt{q}-1}$ for some real $\Gt\in(0,0.5)$, where $\fb=\log q$. Then, for infinitely many $n$, there exist one-point algebraic geometry codes $C$ of rate $1-\Ge$ in $\F_q^n$ such that every coordinate of a given codeword of $C$ can be repaired by the remaining $n-1$ coordinates with bandwidth $(n-1)(1-\Gt)\log q$. 

 In addition, storage in each node for our code is close to the minimum storage. More precisely speaking, our code requires storage $\mathfrak{b}$, while an MSR code with the same length and rate requires storage  $\frac {1-\Ge}{1-\Ge+2^{-\mathfrak{b}/2}}\mathfrak{b}$.}
 \end{quote}
On the other hand, it is easy to see that one requires bandwidth $(1-\Ge)n\log q$ in order to recover the whole codeword. This bandwidth is strictly greater than $(n-1)(1-\Gt)\log q$.

\begin{rmk} Note that our repair in the above main result is a linear repair over $\F_p$ (see \cite[Definition 4]{GW16}). By mimicking the proof of Theorem 3 of \cite{GW16}, one can show that for a $q$-ary code $C$ with dual distance $\Gd^{\perp}$, any linear repair over $\F_p$ has repairing bandwidth $B$ satisfying
\begin{equation}\label{eq:3.10a}B\ge (n-1)\log\left(\frac{n-1}{\Gd^{\perp}-1}\right).\end{equation}
If  applying the above bound to algebraic geometry codes, we get
\begin{equation}\label{eq:3.10b}B\ge (n-1)\log\left(\frac{n-1}{n-k+\g-2}\right),\end{equation}
where $k$ is the dimension of the code and $\g$ is the genus of the underlying curve. Since our codes are based on the Garcia-Stichtenoth tower (see Section 3), we have $\frac{\g}n\rightarrow\frac1{\sqrt{q}-1}$. By setting $n-k=\Ge n$ with $\Ge=2^{(\tau-1/2)\fb}\ge \frac2{\sqrt{q}-1}$ for some real $\Gt\in(0,0.5)$, the bound \eqref{eq:3.10b} gives
\begin{equation}\label{eq:3.10c}B\gtrapprox (n-1)\left(\frac12-\Gt\right)\log q.\end{equation}

Note that the bandwidth given in our main result is roughly twice of the bound \eqref{eq:3.10c}. This means that one can not do much better than our main result for linear repairing.
\end{rmk}
\subsection{Comparison with RS codes}
In \cite{GW16}, it was shown that one can repair an $n$-ary Reed-Solomon code of length $n$ and rate $1-1/p$ with bandwidth $(n-1)\log p$, where $n$ is a prime power and $\F_p$ is a subfield of $\F_n$.

Let us consider both AG codes and RS codes over the same base field $\F_q$ and  put $\Ge=\frac1p=2^{(\tau-1/2)\log n}$ for some real $\Gt\in(0,0.5)$, then one can repair a $q$-ary Reed-Solomon code of length $n=q$ and rate $1-1/p$ with bandwidth $(n-1)(\frac12-\Gt)\log n=(n-1)(\frac12-\Gt)\log q$.
Hence, an AG code with rate $1-\Ge$ has bandwidth roughly  twice the bandwidth required by RS codes with the same rate. However, the code length $n$ for the AG codes is unbounded, while the code length of the RS codes is upper bounded by $q$. This can be viewed as a trade-off between bandwidth and code length constraint.

If we compare   two classes of codes by fixing the same length $n$, i.e., AG codes are defined over a field field $\F_q$ for a constant $q$, while RS codes are defined over $\F_n$ (where $n$ is much bigger than $q$). We set $\Ge=\frac1p=2^{(\tau-1/2)\log q}$ and it turns out that with the same rate, RS codes have the bandwidth $(n-1)(\frac12-\Gt)\log q$. This implies that with the same number of nodes, AG codes can store smaller size of data with bandwidth at most twice the bandwidth required by RS codes with the same rate. Furthermore, AG codes are nearly MSR codes.

One may argue that,  if setting $p=2$, one obtains  RS codes of rate $1/2$ with repairing bandwidth $(n-1)$. In this case, we take $q=25$ and $\Gt=\frac 12-\frac1{2\log5}$. Then $\Ge=2^{(\Gt-1/2)\log25}=\frac 12$. Therefore, one get an AG code of rate $\frac 12$ and repairing bandwidth $(n-1)(\log 5+\frac 12)\approx2.828(n-1)$. Although the repairing bandwidth of our AG code is $2.8$ times of the RS codes, our code alphabet $q$ is a constant, i.e., $q=25$.

\subsection{Our approach}

Let us summarize the approach used in Guruswami-Wootters' paper \cite{GW16}  as follows. Fix a subfield $\F_p$ of $\F_q$ with $t=[\F_q:\F_p]$. Choose a basis $\zeta_1,\dots,\zeta_t$ of $\F_q$ over $\F_p$.
For each pair $(i,u)\in[n]\times[t]$, find a polynomial $h_{(i,u)}$ such that $\deg(h_{(i,u)})\le n-k-1$ and  $b_j^{(i)}=\dim_{\F_p}\Span_{\F_p}\{h_{(i,u)}(P_j):\; u=1,2,\dots,t\}$ satisfies $b_i^{(i)}=t$ and $b_j^{(i)}<t$ for all $1\le i\neq j\le n$. The function $h_{(i,u)}$ chosen in Guruswami-Wootters' paper \cite{GW16} is $\frac{\Tr(\zeta_u(x-a_i))}{x-a_i}$, where $\Tr$ is the trace function from $\F_q$ to $\F_p$ and $a_i$ is the evaluation point at position $i$. Assume that $(f(a_1),\dots,f(a_n))$ is the codeword stored in the $n$ nodes with $a_j\in\F_q$ and $\deg(f)\le k-1$.  Thus, one can download $\sum_{j\in[n]\setminus\{i\}}b_j^{(i)}\log p$ bits in total to recover $f(a_i)$.

To generalize the above idea from RS codes to AG codes, the choice of the function $h_{(i,u)}$ is the key part. One can not simply take $h_{(i,u)}$ to be $\frac{\Tr(\zeta_uz)}{z}$ for some function $z$ with $z(P_i)=0$. Otherwise, the  degree of the pole divisor of $\frac{\Tr(\zeta_uz)}{z}$ is $(p^{t-1}-1)\deg((z)_{\infty})$, where $(z)_{\infty}$ is the pole divisor of $z$. Thus, the number $(p^{t-1}-1)\deg((z)_{\infty})$ is too big to find a suitable  codeword in the dual code. Instead of the trace function, we choose a $p$-linearized function $L$ of degree $p^{l}$  and let $h_{(i,u)}=\frac{L(\zeta_uz)}{z}$ so that the pole divisor of $\frac{L(\zeta_uz)}{z}$ has degree $(p^l-1)\deg((z)_{\infty})$. By adjusting $l$, we can control this degree and find a suitable codeword in the dual code.

\subsection{Organization}
In Section \ref{sec:2}, we will introduce some preliminaries including dual basis of finite fields, function fields, algebraic geometry codes and their dual codes, etc. In Section \ref{sec:3}, we present  an efficient repair decoding algorithm for arbitrary one-point algebraic geometry codes and apply it to the algebraic geometry codes based on the Garcia-Stichtenoth tower. The  repairing of Reed-Solomon codes and Hermitian codes is discussed under the framework of algebraic geometry codes in the last section.
\section{Preliminaries}\label{sec:2}
In this section, we discuss some algebraic backgrounds that are needed for the next section.
\subsection{Background on finite fields}\label{subsec:2.1}
Throughout this paper, we assume that $\F_q$ is the finite field with even characteristic. We also assume that $\F_q/\F_p$ is a field extension with $[\F_q:\F_p]=t$. For an $\F_p$-basis $\{\zeta_1,\zeta_2,\dots,\zeta_t\}$ of $\F_q$, its dual basis is an $\F_p$-basis $\{\theta_1,\theta_2,\dots,\theta_t\}$ of $\F_q$ satisfying
\begin{equation}\label{eq:1}
\Tr(\zeta_i\theta_j)=\left\{\begin{array}{ll}
1&\mbox{if $i=j$}\\
0&\mbox{if $i\neq j$,}
\end{array}
\right.
\end{equation}
where $\Tr$ is the trace map from $\F_q$ to $\F_p$. By \cite[Section 3 of Chapter 2]{LN}, we know existence of dual basis for any given $\F_p$-basis. Throughout this paper, we will fix an $\F_p$-basis $\{\zeta_1,\zeta_2,\dots,\zeta_t\}$ and its dual $\{\theta_1,\theta_2,\dots,\theta_t\}$. For any element $\Ga\in\F_q$, let $\Ga=\sum_{i=1}^ta_i\theta_i$ with $a_i\in\F_p$. Then for any $j\in[n]$, $\Tr(\Ga\zeta_j)=\Tr\left(\sum_{i=1}^ta_i\theta_i\zeta_j\right)=\sum_{i=1}^ta_i\Tr(\theta_i\zeta_j)=a_j$. This implies that
\begin{equation}\label{eq:2}
\Ga=\sum_{j=1}^t\Tr(\Ga\zeta_j)\theta_j.
\end{equation}
The above trace representation will be used to calculate bandwidth.

For a set $\{\Gb_1,\Gb_2,\dots,\Gb_n\}$ of $\F_q$, we denote by $\Span_{\F_p}\{\Gb_1,\Gb_2,\dots,\Gb_n\}$ the $\F_p$-vector space generated by $\{\Gb_1,\Gb_2,\dots,\Gb_n\}$. Note that, for any nonzero element $\Ga\in\F_q$,  $\Span_{\F_p}\{\Gb_1,\Gb_2,\dots,\Gb_n\}$ and $\Span_{\F_p}\{\Ga\Gb_1,\Ga\Gb_2,\dots,\Ga\Gb_n\}$ have the same dimension over $\F_p$.

Let $V$ be an $\F_p$-subspace of $\F_q$. We define a $p$-linearized polynomial $L_V(x):=\prod_{\Ga\in V}(x-\Ga)$. Then $L_V(x)$ defines an $\F_p$-linear map from $\F_q$ to $\F_q$ given by $\Gb\mapsto L_V(\Gb)$. Apparently, the kernel of $L_V(x)$ is $V$ and hence the image space $\Im(L_V)$ has $\F_p$-dimension equal to $\log_p q-\dim_{\F_p}(V)$.

\subsection{Background on algebraic function fields}

Let $F$ be a function field  of genus $\g$ defined over $\F_q$ (we assume that $\F_q$ is the full constant field of $F$, i.e., the algebraic closure of $\F_q$ in $F$ is $\F_q$).  An element of $F$ is called a function.  The normalized discrete valuation corresponding to a place $P$ of $F$ is written as $\nu_P$. A place $P$ is said $\F_q$-rational if it has degree one. Let $\PP_F$ denote the set of places of $F$. A divisor $G$ of $F$ is of the form $\sum_{P\in \PP_F}m_pP$ with finitely many nonzero coefficients $m_P$. The degree of $G$ is defined by  $\deg(G)=\sum_{P\in\PP_F}m_p\deg(P)$. The support of $G$, denoted by  $\supp(G)$, is defined to be $\{P\in\PP_F: \; m_P\neq 0\}$. Thus $\supp(G)$ is a finite set. It is clear that all divisors form a free abelian group generated by $\PP_F$. For a nonzero function $f$, the principal devisor of $f$ is defined to be $(f):=\sum_{P\in\PP_F}\nu_P(f)P$. Then the degree of a principal divisor is zero. Two divisors $G_1$ and $G_2$ are are said equivalent if there exists a nonzero function $f$ such that $G_1=(f)+G_2$.

We denote by $\GO_F$ the set of differentials of $F$. All differentials of $F$ form a vector space of dimension one over $F$.  This means that if $t\in F$ satisfies $dt\neq0$, then $\GO_F=Fdt$. The degree of the divisor $(dt)$ is $2\g-2$. Hence, for every nonzero element $fdt$ with $f\in F\setminus\{0\}$, the degree of $(fdt)$ is $\deg(f)+\deg(dt)=2\g-2$. A divisor $(fdt)$ with $f\in F\setminus\{0\}$ is called a canonical divisor of $F$. It is clear that every  canonical divisor has degree $2\g-2$ and all canonical divisors are equivalent.

For a divisor $G$, one can define two spaces
\[\mL(G)=\{f\in F\setminus\{0\}:\quad (f)+G\geq0\}\cup\{0\}\]
and
\[\GO(G)=\{\Go\in\GO_F\setminus\{0\}:\; (\Go)\ge G\}\cup\{0\}.\]
Then both $\mL(G)$ and $ \GO(G)$ are finitely dimensional spaces over $\F_q$. Dimensions of $\mL(G)$ and $ \GO(G)$ are denoted by $\ell(G)$ and $i(G)$, respectively. One has the  identity $i(G)=\ell(K-G)$, where $K$ is a canonical divisor. Then the Riemann-Roch theorem says that
\[\ell(G)=\deg(G)-\g+1+i(G)=\deg(G)-\g+1+\ell(K-G),\]
where $K$ is a canonical divisor. Consequently $\ell(G)\ge \deg(G)-\g+1$ and the equality holds if $\deg(G)>2\g-2$. The reader may refer \cite{St93} for the details.

\subsection{Algebraic geometry codes}\label{subsec:2.3}
Let $F/\F_q$ be an algebraic function field. Assume that $F$ has $n$ distinct $\F_q$-rational places $\{ P_1,P_2,\dots,P_n\}$.  Denote by $\mP$ the set $\{P_1,P_2,\dots,P_n\}$. Let $G$ be a divisor of $F$ with $\supp(G)\cap\mP=\emptyset$.  Consider the  algebraic geometry  code defined by
Goppa \cite{G81}:
\begin{equation}
 C_L(G,\mP):=\{(f(P_1),f(P_2),\dots, f(P_n)):\; f\in\mL(G)\}.
\end{equation}
 The code $C_L(G,\mP)$ is called a functional algebraic geometry code.

Now we define another code $C_\Omega(G,\mP)$ by
\begin{equation}
C_\Omega(G,\mP)=\left\{(\res_{P_1}(\Go),\dots,\res_{P_n}(\Go)):\;\Go\in
\Omega\left(G-\sum_{i=1}^nP_i\right)\right\},\end{equation}
where $\res_{P_i}(\Go)$ stands for the residue of $\Go$ at $P_i$.

The code $C_{\GO}(G,\mP)$ is called a differential algebraic geometry code.  It was proved in \cite{G81} that
 $C_\Omega(G ,\mP)$ is the Euclidean dual of $C_L(G,\mP)$.

We have the following results.
\begin{prop}\label{prop:2.1} The dual code of $C_L(G,\mP)$ is $C_{\GO}(G,\mP)$ and both the codes have length $n$.  Furthermore, we have
\begin{itemize}
\item[{\rm (i)}] The dimension $k$ of $C_L(G,\mP)$ is $\ell(G)-\ell\left(G-\sum_{i=1}^nP_i\right)$ and the dimension $k^{\perp}$ of $C_{\GO}(G,\mP)$ is $i(G-\sum_{i=1}^nP_i)-i(G)$.
\item[{\rm (ii)}] If $\deg(G)$ satisfies $\deg(G)<n$, then $k=\ell(G)\geq \deg(G)-\g+1.$
\item[{\rm (iii)}] If $\deg(G)$ satisfies $\deg(G)>2\g-2$, then $k^{\perp}=i(G-\sum_{i=1}^nP_i)\geq n-\deg(G)+\g-1.$
\item[{\rm (iv)}] If additionally $2\g-2<\deg(G)<n$, then $k=\deg(G)-\g+1$ and $k^{\perp}=n-\deg(G)+\g-1.$
\end{itemize}
\end{prop}
 Note that the dimension of $C_L(G,\mP)$ is in general less than the degree $\deg(G)$ of $G$.
Again the reader may refer \cite{St93} for the details.

\subsection{Dual codes of algebraic geometry codes}
From subsection \ref{subsec:2.3}, we know that the dual code of $C_L(G,\mP)$ is $C_\Omega(G,\mP)$. In this subsection, we are going to show that the dual code $C_\Omega(G,\mP)$ of $C_L(G,\mP)$ contains a codeword with zero and nonzero coordinates at some specific positions.

\begin{prop}\label{prop:2.2} Let $G$ be a divisor satisfying $\deg(G)<d$ and $\supp(G)\cap\mP=\emptyset$. Then for any $i\in[n]$ and subset $S\subseteq[n]\setminus\{i\}$ with $|S|=d$, there exists a codeword $(\res_{P_1}(\Go),\dots,\res_{P_n}(\Go))\in C_{\GO}(G,\mP)$ such that $\res_{P_i}(\Go)\neq0$ and $\res_{P_j}(\Go)=0$ for all $j\in [n]\setminus(S\cup\{i\})$.
\end{prop}
\begin{proof}
As $i(G-\sum_{l\in S}P_l)=|S|-\deg(G)+2\g-2-\g+1=d-\deg(G)+\g-1$ and $i(G-\sum_{l\in S}P_l-P_i)=|S|+1-\deg(G)+2\g-2-\g+1=d-\deg(G)+\g$, there exists a differential $\Go\in\GO(G-\sum_{l\in S}P_l-P_i)\setminus\GO(G-\sum_{l\in S}P_l)$. This implies that $\res_{P_i}(\Go)\neq0$ and $\res_{P_j}(\Go)=0$ for all $j\in [n]\setminus(S\cup\{i\})$. Furthermore, we have $(\res_{P_1}(\Go),\dots,\res_{P_n}(\Go))\in C_{\GO}(G,\mP)$. The proof is completed.
\end{proof}

\subsection{Garcia-Sticthenoth tower}\label{subsec:2.5}
There are two asymptotically optimal towers introduced by Garcia-Sticthenoth \cite{GS95,GS96}. We adopt the tower given in \cite{GS96}.

Let $q$ be a perfect square. The tower $\{F_e\}_{e=1}^{\infty}$ is recursively defined as follows: $F_1=\F_q(x_1)$ is a rational function field with variable $x_1$ and $F_e=\F_q(x_1,\dots,x_e)$, where $x_i$ satisfy the recursive equations:
  \[x_{i+1}^{\sqrt{q}}+x_{i+1}=\frac{x_i^{\sqrt{q}}}{x_i^{\sqrt{q}-1}+1}\qquad \mbox{for $i=1,2,\dots,e-1$}.\]
A careful analysis shows that the genus $\g(F_e)\le q^{e/2}$. Moreover, for any  $\Ga\in\F_q\setminus\{ \Gb\in\F_q : \Gb^{\sqrt{q}} + \Gb = 0\}$, the zero place of $x_1-\Ga$ in the rational
function field $F_1 =\F_q(x_1)$ splits completely in all extensions $F_e/F_1$. The infinite place
of $F_1$ is totally ramified in all extensions $F_e/F_1$. Hence the number of rational places of
$F_e$ satisfies
\[N(F_e)\ge q^{e/2}(\sqrt{q}-1)+1.\]
Let us denote by $\Pin$ the unique place of $F_e$ lying on the infinite place of $F_1$ and denote by $\mP$ those places of $F_e$ lying on $x_1-\Ga$ for all $\Ga\in\F_q\setminus\{ \Gb\in\F_q : \Gb^{\sqrt{q}} + \Gb = 0\}$.

It was shown in \cite{Shum01} that one can construct a basis of $\mL((m-1)\Pin)$ in $O(m^3)$ operations of $\F_q$ and the algebraic geometry code $C_L(G,\mP)$ can be constructed in $O(m^3)$ operations of $\F_q$ as well if $n=O(m)$.
\section{Repairing AG codes}\label{sec:3}
In this section, we consider  repairing of algebraic geometry codes. As our repairing algorithm works well for one-point algebraic geometry codes, we only consider one-point algebraic geometry codes from now onwards. Let $F/\F_q$ be an algebraic function field.
 Assume that $F$ has $n+1$ distinct $\F_q$-rational places $\{\Pin, P_1,P_2,\dots,P_n\}$.  Denote by $\mP$ the set $\{P_1,P_2,\dots,P_n\}$.  Consider the algebraic geometry code $C_L((m-1)\Pin,\mP)$ defined in Subsection \ref{subsec:2.3}. Let  $\{\zeta_1,\zeta_2,\dots,\zeta_t\}$ be an $\F_p$-basis of $\F_q$ and let $\{\theta_1,\theta_2,\dots,\theta_t\}$ be its dual.
Let us first show that any $m$ coordinates can repair the whole codeword for $C_L((m-1)\Pin,\mP) $.
\begin{lem}\label{lem:3.1} For every codeword $(f(P_1),f(P_2),\dots,f(P_n))$ of $C_L((m-1)\Pin,\mP)$, every $i\in[n]$ and every subset $S$ of $[n]\setminus\{i\}$ of size $m$, $f(P_i)$ can be repaired by the set $\{f(P_j):\; j\in S\}$.
\end{lem}
\begin{proof} By Proposition \ref{prop:2.2}, one can find a differential $\Go\in \GO((m-1)\Pin-\sum_{j\in S}P_j-P_i)\setminus\GO((m-1)\Pin-\sum_{j\in S}P_j)$. This implies that $\res_{P_i}(\Go)\neq0$ and $\res_{P_u}(\Go)=0$ for all $u\in[n]\setminus (S\cup\{i\})$. As $(\res_{P_1}(\Go),\dots,\res_{P_n}(\Go))$ belongs to the dual code  $C_{\GO}((m-1)\Pin,\mP)$ of  $C_{L}((m-1)\Pin,\mP)$, we have
\[\res_{P_i}(\Go)f(P_i)=-\sum_{j\in[n]\setminus{i}}\res_{P_j}(\Go)f(P_j)=-\sum_{j\in S}\res_{P_j}(\Go)f(P_j).\]
The desired result follows since $\res_{P_i}(\Go)\neq0$.
\end{proof}
The above lemma shows that one can repair $f(P_i)$ by  downloading the whole data from any other $m$ nodes. Lemma \ref{lem:3.1} also follows from that fact that $C_L((m-1)\Pin,\mP) $ has $(n-\Gd+1)$-reconstruction, where $\Gd$ is the minimum distance of $C_L((m-1)\Pin,\mP) $ and hence $n-\Gd\ge n-(n-m+1)+1=m$.

 The next lemma  shows that we can download partial data from any $d$ nodes to repair $f(P_i)$.

\begin{lem}\label{lem:3.2} Let $m,r,d$ be  positive integers satisfying $m\le d-r$.  Assume that for each pair $(i,u)\in[n]\times[t]$, there exists a function $h_{(i,u)}\in \mL(r\Pin)$ such that $h_{(i,u)}(P_i)=\zeta_u$. Then  $C_L((m-1)\Pin,\mP)$ is a weak $[n, m, d]$-regenerating code having the secondary parameters $(\mathfrak{b}:=\log q, M:=\ell((m-1)\Pin)\log q, B)$ with bandwidth
$B=\max_{i\in[n], S\subseteq [n]\setminus\{ i\}, |S|=d}\sum_{j\in S}b_j^{(i)}\log p$, where
\begin{equation}
b_j^{(i)}=\dim_{\F_p}\Span_{\F_p}\{h_{(i,u)}(P_j):\; u=1,2,\dots,t\}.
\end{equation}
\end{lem}
\begin{proof} For a codeword $(f(P_1),f(P_2),\dots,f(P_n))\in C_L((m-1)\Pin,\mP)$ with $f\in\mL((m-1)\Pin)$,  assume that we are going to repair $f(P_i)$. Let $S\subseteq [n]\setminus\{ i\}$ with $|S|=d$.   By Proposition \ref{prop:2.2}, there exists a codeword $(\res_{P_1}(\Go),\dots,\res_{P_n}(\Go))\in C_{\GO}((r+m-1)\Pin,\mP)$ such that $\res_{P_i}(\Go)\neq0$ and $\res_{P_l}(\Go)=0$ for all $l\in [n]\setminus(S\cup\{i\})$.

For each  $j\in S$, let $J\subseteq [t]$ with $|J|=b_j^{(i)}$ such that the set $\{h_{(i,v)}(P_j):\; v\in J\}$ is an $\F_p$-basis of $\Span_{\F_p}\{h_{(i,u)}(P_j):\; u=1,2,\dots,t\}$.  We download the following data from the node storing $f(P_j)$:
\begin{equation}\label{eq:6}
\Tr\left(\frac{\res_{P_j}(\Go)}{\res_{P_i}(\Go)}\times h_{(i,u)}(P_j)f(P_j)\right), \quad \mbox{for all $v\in J$}.
\end{equation}
This means that one needs to download at most $\sum_{j\in S}b_j^{(i)}\log p$ bits of data in total. Now we show that with all data in \eqref{eq:6} for all $j\in S$, we can repair $f(P_i)$.

First of all, for each $1\le u\le t$,  $h_{(i,u)}(P_j)$ is an $\F_p$-linear combination of $\{h_{(i,v)}(P_j):\; v\in J\}$, i.e., there exists a set $\{\Gl_v\}_{v\in J}$ of elements of $\F_p$ such that $h_{(i,u)}(P_j)=\sum_{v\in J}\Gl_vh_{(i,v)}(P_j)$. This gives
\[\Tr\left(\frac{\res_{P_j}(\Go)}{\res_{P_i}(\Go)}\times h_{(i,u)}(P_j)f(P_j)\right)=\sum_{v\in J}\Gl_v\Tr\left(\frac{\res_{P_j}(\Go)}{\res_{P_i}(\Go)}\times h_{(i,v)}(P_j)f(P_j)\right).\]
This implies that, for all $1\le u\le t$, one can compute $\Tr\left(\frac{\res_{P_j}(\Go)}{\res_{P_i}(\Go)}\times h_{(i,u)}(P_j)f(P_j)\right)$  from the downloaded data given in \eqref{eq:6}.

Now by \eqref{eq:2}, it is sufficient to know $\Tr(f(P_i)\zeta_u)$ for all $u=1,2,\dots,t$.
Since $(\res_{P_1}(\Go),\dots,\res_{P_n}(\Go))\in C_{\GO}((r+m-1)\Pin,\mP)$ and $(h_{(i,u)}(P_1)f(P_1),\dots,h_{(i,u)}(P_n)f(P_n))\in C_L((r+m-1)\Pin,\mP)$, we have \[0=\sum_{j=1}^n \res_{P_j}(\Go)h_{(i,u)}(P_j)f(P_j)=\sum_{j\in S\cup\{i\}} \res_{P_j}(\Go)h_{(i,u)}(P_j)f(P_j).\]
 This gives the following identity
\begin{equation}
\Tr(f(P_i)\zeta_u)=\Tr(f(P_i)h_{(i,u)}(P_i))=-\sum_{j\in S}\Tr\left(\frac{\res_{P_j}(\Go)}{\res_{P_i}(\Go)}\times h_{(i,u)}(P_j)f(P_j)\right).
\end{equation}
The desired result follows.
\end{proof}

By explicitly constructing functions $h_{(i,u)}$ given in Lemma \ref{lem:3.2}, we obtain the following regenerating codes.

\begin{theorem}\label{thm:3.3}  Let $F/\F_q$ be a function field over $\F_q$ with genus $\g$ and $n+1$ distinct $\F_q$-rational places $\{\Pin, P_1,P_2,\dots,P_n\}$.  Denote by $\mP$ the set $\{P_1,P_2,\dots,P_n\}$. If $2\g\le m\le d-(p^l-1)(\g+1)$, then  $C_L((m-1)\Pin,\mP)$ is a weak $[n, m, d]$-regenerating code having the secondary parameters $(\mathfrak{b}:=\log q, M:=(m-\g)\log q, B)$ with
$B=d\log q-(d-\g)l\log p$.
\end{theorem}
\begin{proof} First of all, note that, in this case, we have $\ell((m-1)\Pin)=m-\g$.

Choose an $\F_p$-subspace $V$ of $\F_q$ of dimension $l$. Consider the $\F_p$-linear map $L_V(x)=\prod_{\Ga\in V}(x-\Ga)$ of $\F_q$ defined in Subsection \ref{subsec:2.1}. For each $i\in[n]$, choose a nonzero function $h_i\in \mL((\g+1)\Pin-P_i)$ (note that this is possible since $\ell((\g+1)\Pin-P_i)\ge \g-\g+1=1$) and put
\begin{equation}\label{eq:8}
h_{(i,u)}:=\frac{L_V(\zeta_u h_i)}{h_i}.
\end{equation}
It is clear that $h_{(i,u)}$ is a nonzero function in $\mL((p^l-1)(\g+1)\Pin)$ and $h_{(i,u)}(P_i)=\zeta_u$.

Let $S\subseteq [n]\setminus\{ i\}$ with $|S|=d$. Denote by $I_i$ the set $\{j\in S:\; h_i(P_j)=0\}$. Then we have $h_i\in \mL((\g+1)\Pin-P_i-\sum_{j\in I_i}P_j)$. This gives that $|I_i|\le \g$.

If $j\in S\setminus I_i$, we have
\[\Span_{\F_p}\{h_{(i,u)}(P_j):\; u=1,2,\dots,t\}=\Span_{\F_p}\left\{\frac{L_V(\zeta_u h_i(P_j))}{h_i(P_j)}:\; u=1,2,\dots,t\right\}\subseteq \frac{1}{h_i(P_j)}L_V(\F_q).\]
Thus, $b_j^{(i)}\le\dim_{\F_p}(L_V(\F_q))=\log_p q-l$.

If $j\in I_i$, we have a trivial bound $b_j^{(i)}\le \log_p q$. By Lemma \ref{lem:3.2}, the bandwidth is upper bounded by
\[\sum_{j\in S}b_j^{(i)}\log p=\sum_{j\in I_i}b_j^{(i)}\log p+\sum_{j\in S\setminus I_i}b_j^{(i)}\log p\le \g\log q+(d-\g)(\log q-l\log p).\]
\end{proof}
By applying Theorem \ref{thm:3.3} to the Garcia-Stichtenoth tower given in Subsection \ref{subsec:2.5}, we obtain the following result.
\begin{theorem}\label{thm:3.4} Let $q$ be a perfect square and let $e\ge 1$ be an integer. Put $n=q^{e/2}(\sqrt{q}-1)$. If $q$ is a power of $p$, $l\le \log_pq$ and $2q^{e/2}\le m\le d-(p^l-1)(q^{e/2}+1)$, then  there is a weak $[n, m, d]$-regenerating code having the secondary parameters $(\mathfrak{b}:=\log q, M:=(m-q^{e/2})\log q, B)$ with
$B=d\log q-(d-q^{e/2})l\log p$. 

In addition, storage in each node for our code is close to the minimum storage. More precisely speaking, our code requires storage $\log q$, while an MSR  code with the same length and rate requires storage  $\frac{R}{R+1/\sqrt{q}}\log q$, where $R$ is the rate of the code.
\end{theorem}
\begin{proof} Consider the function field $F_e$ of  the Garcia-Stichtenoth tower given in Subsection \ref{subsec:2.5}. Let us denote by $\Pin$ the unique place of $F_e$ lying on the infinite place of $F_1$ and denote by $\mP$ those places of $F_e$ lying on $x_1-\Ga$ for all $\Ga\in\F_q\setminus\{ \Gb\in\F_q : \Gb^{\sqrt{q}} + \Gb = 0\}$. By Theorem \ref{thm:3.3}, the code $C_L((m-1)\Pin,\mP)$ is a weak $[n, m, d]$-regenerating code having the secondary parameters $(\log q, (m-\g)\log q, B)$ with
$B=d\log q-(d-\g)l\log p$.

The code $C_L((m-1)\Pin,\mP)$ can be constructed in $O(n^3\log q)$ bit operations. Furthermore, both the function $h_i\in\mL((\g+1)\Pin-P_i)$ in Theorem \ref{thm:3.3} and the differential $\Go\in\GO((m-1+(p^l-1)(q^{e/2}+1))\Pin-\sum_{l\in S}P_l-P_i)\setminus\GO((m-1+(p^l-1)(q^{e/2}+1))\Pin-\sum_{l\in S}P_l)$ in Proposition \ref{prop:2.2} can be explicitly constructed in $O(n^3\log q)$ bit operations \cite{Shum01}. The desired result on complexity follows.

Clearly, storage in each node for our code is $\log q$. Let $R=(m-q^{e/2})/n$ be the rate of our code. Then $m=Rn+q^{e/2}$. A minimum storage regenerating code with the same length $n$ and rate $R$ requires storage
\[\frac{M}{m}=\frac{Rn\log q}{Rn+q^{e/2}}=\frac{R}{R+1/\sqrt{q}}\log q.\]
\end{proof}
Set $l=1$ and $m-q^{e/2}=n(1-\Ge)$. Then the inequalities in Theorem \ref{thm:3.4} forces $\frac{p}{\sqrt{q}-1}\le \Ge\le 1-\frac{1}{\sqrt{q}-1}$. Thus, we obtain the following result by Theorem \ref{thm:3.4}.
\begin{cor}\label{cor:3.5}
Let $q$ be a perfect square and let $\mathfrak{b}=\log q$. Let $\Ge$ be a real in the interval $(\frac{p}{\sqrt{q}-1}, 1-\frac{1}{\sqrt{q}-1}$. Then, for infinitely many $n$, there exist codes $C$ of dimension $n(1-\Ge)$ in $\F_{q}^n$ such that every coordinate of a given codeword of $C$ can be repaired by the remaining $n-1$ coordinates with bandwidth $(n-1)(\frac{\mathfrak{b}}2+\log\frac1{\Ge})$. 

In addition, storage in each node for our code is close to the minimum storage. More precisely speaking, our code requires storage $\mathfrak{b}$, while a a minimum storage code with the same length and rate requires storage  $\frac {1-\Ge}{1-\Ge+2^{-\mathfrak{b}/2}}\mathfrak{b}$.
\end{cor}

\section{Regenerating codes from Reed-Solomon codes and Hermitian codes}\label{sec:4}
The reason why we got weak regenerating codes in Section \ref{sec:3} is that the polynomial $h_i$ may have some zeros apart from $P_i$. If we consider some function fields such as rational function fields and Hermitian function fields, this can be avoided. In other words, we can get strong   regenerating codes with smaller bandwidth for Reed-Solomon codes and Hermitian codes.
\subsection{Reed-Solomon codes}
Local repairing of Reed-Solomon codes was considered in \cite{GW16}. In this subsection, we revisit it under the framework of Section \ref{sec:3}.

If $F$ is the rational function field, then the corresponding algebraic geometry codes are actually Reed-Solomon codes. Let $\Pin$ be the unique pole of $x$.
\begin{theorem}\label{thm:4.1}
If $m\le d-p^l+1$, then the Reed-Solomon code $C_L((m-1)\Pin,\mP)$ is a strong $[n, m, d]$-regenerating code having the secondary parameters $(\mathfrak{b}:=\log q, M:=m\log q, B)$ with bandwidth
$B=d\log q-dl\log p$.
\end{theorem}
\begin{proof} Note that the genus $\g=0$ in this case.
Thus, $I_i$ in the proof of Theorem \ref{thm:3.3} is the $\emptyset$. Therefore, we download the same bits from each node and the code is a strong regenerating code. In addition, for $i\in[n]$ and any subset $S\subset[n]\setminus\{i\}$ with $|S|=d$, the condition that $m\le d-p^l+1$ implies that the function $h_{(i,u)}$ defined in \eqref{eq:8} is a function  in $\mL((p^l-1)\Pin)$. By using the same arguments in Lemma \ref{lem:3.2} and Theorem \ref{thm:3.3}, we obtain the desired result.
\end{proof}
By setting $n=q$, $d=n-1$ and $l=\log_pq-1$ in Theorem \ref{thm:4.1}, we get the same result given in \cite[Theorem 1]{GW16}.
\begin{cor}\label{cor:4.2} Let $q$ be a power of $p$ and $n=q$. Then for any $m\le n(1-\frac1p)$, the Reed-Solomon code of length $n$ and rate at most $1-\frac 1p$ is a strong $[n, m, n-1]$-regenerating code having the secondary parameters $(\mathfrak{b}:=\log q, M:=m\log q, B)$ with bandwidth
$B=(n-1)\log p$.
\end{cor}
\subsection{Hermitian codes}\label{sec:5}
Let $q=r^2$. The Hermitian curve over $\F_q$ is defined by
\begin{equation}
y^r+y=x^{r+1}.
\end{equation}
We denote by $\mathrm{H}$ the function field $\F_q(x,y)$ of the above curve.

 {\bf Rational points:}
 This curve has $r^3+1$ $\F_q$-rational points in total. One point is the common pole of $x$ and $y$ and it is called ``point at infinity". We denote it by $\Pin$. The other $r^3$ ``finite" $\F_q$-rational points are given by $\{(\Ga,\Gb)\in\F_q^2:\; \Gb^r+\Gb=\Ga^{r+1}\}.$  Note that for any given $\Ga\in\F_q$, there are exactly $r$ elements $\Gb\in\F_q$ such that $\Gb^r+\Gb=\Ga^{r+1}$. Thus, one can see that there are exactly $r^3$ pairs $(\Ga,\Gb)\in\F_q^2$ satisfying $\Gb^r+\Gb=\Ga^{r+1}$. Let $Z_0:=\{\Gg\in\F_q:\; \Gg^r+\Gg=0\}$. Then $Z_0$ is a subset of $\F_q$ with cardinality $r$. For each $\Gb\in Z_0$, $(0,\Gb)$ is a point on the Hermitian curve. For each $\Gb\in \F_q\setminus Z_0$, there are exactly $r+1$ elements $\Ga\in\F_q$ satisfying $\Gb^r+\Gb=\Ga^{r+1}$. This counting also gives $r^3$ ``finite" $\F_q$-rational points in total.

  {\bf Principal divisor:} Denote by $P_{\Ga,\Gb}$ the point  $(\Ga,\Gb)\in\F_q^2$ satisfying $\Gb^r+\Gb=\Ga^{r+1}$. Let $\mP$ be the collection of all ``finite" $r^3$ rational points on $\mathrm{H}$.  For each $\Ga\in\F_q$, the principal divisor $(x-\Ga)$ is $\sum_{\Gb^r+\Gb=\Ga^{r+1}}P_{\Ga,\Gb}-r\Pin$ and $x^q-x=\sum_{P\in\mP}P-r^3\Pin$.
 For an element $\Ga\in\F_q$, denote by $Z_{\Ga}$ the set $Z_{\Ga}:=\{\Gg\in\F_q:\; \Gg^g+\Gg=\Ga\}$. The following lemma can be found in \cite{X95}.
 \begin{lem}\label{lem:2.1} Let $\Ga\in\F_q$. Then we have
 \begin{itemize}
 \item[{\rm (i)}] $(y+\Ga x-\Gg)=(r+1)(P_{-\Ga^r, \Ga^{r+1}+\Gg}-\Pin)$ if $\Gg\in Z_{-\Ga^{r+1}}$.
  \item[{\rm (ii)}] $(y+\Ga x-\Gg)=\sum_{i=1}^{r+1}R_i-(r+1)\Pin$ for some distinct rational points $R_1,R_2,\dots,R_{r+1}$ if $\Gg\in \F_q\setminus Z_{-\Ga^{r+1}}$.
 \end{itemize}
 \end{lem}
 From Lemma \ref{lem:2.1}, we can show that for every $\F_q$-rational point $P\in\mP$, the divisor $(r+1)(P-\Pin)$ is a principal divisor.
 \begin{cor}\label{cor:2.2}
 For every $\F_q$-rational point $P\in\mP$, the divisor $(r+1)(P-\Pin)$ is a principal divisor.
 \end{cor}
 \begin{proof}
Let $P=P_{a,b}\in\mP$, then one can find $\Ga\in\F_q$ such that $a=-\Ga^r$. Then $a^{r+1}=(-\Ga^r)^{r+1}$. Put $\Gg=b-\Ga^{r+1}$. Then it is easy to verify that $\Gg\in Z_{-\Ga^{r+1}}$. By Lemma \ref{lem:2.1}(i), this means that $(y+\Ga x-\Gg)=(r+1)(P_{a,b}-\Pin)$.
\end{proof}

  {\bf Genus, differential and Riemann-Roch space:} The genus of $\mathrm{H}$ is $\g:=r(r-1)/2$.

For $m\ge 0$, we consider the Riemann-Roch space $\mL((m-1)\Pin)$ which is given by
\begin{equation}
\mL((m-1)\Pin):=\Span_{\F_q}\{x^iy^j:\; iq+j(q+1)\le m-1\},
\end{equation}
where $\Span_{\F_q}$ stands for $\F_q$-linear span. The $\F_q$-dimension  of $\mL((m-1)\Pin)$ is given by the cardinality of the set $\{(i,j)\in\ZZ^2_{\ge 0}:\; ir+j(r+1)\le m-1, j\le r-1\}$. By  Riemann-Roch Theorem, one has $\dim_{\F_q}\mL((m-1)\Pin)=m-r(r-1)/2$ if $m\ge 2\g$.

The differential $\eta=\frac{dx}{x^q-x}$ gives the canonical divisor $(\eta)=(r^3+2\g-2)\Pin-\sum_{P\in\mP}P$. Furthermore, we have $\res_P(\eta)=1$ for all $P\in\mP$. By \cite{St93}, we have $C_L((m-1)\Pin,\mP)^{\perp}=C_{\GO}((r^3-m+2\g-1)\Pin,\mP)$ for any $0\le m\le r^3+2\g-2$.

Let $n\le r^3$ and fix a subset $\mS=\{P_1,P_2,\dots,P_n\}$ of $\mP$. We consider local repairing of the Hermitian code $C_L(m\Pin,\mS)$.
We also fix three positive integers $d, m$ satisfying $m\le d\le n-1$.

\begin{theorem}\label{thm:4.5} If $m\le d-(p^l-1)(r+1)$, then $C_L((m-1)\Pin,\mS)$ is a strong $[n,m,d]$-regenerating code with the secondary parameters $(\mathfrak{b}=\log q, M, B)$, where $M=\ell((m-1)\Pin)\log q$ and bandwidth
\begin{equation}B=d(\log q-l\log p).
\end{equation}
\end{theorem}
\begin{proof} Choose an $\F_p$-subspace $V$ of $\F_q$ of dimension $l$. Consider the $\F_p$-linear map $L_V(x)=\prod_{\Ga\in V}(x-\Ga)$ of $\F_q$ defined in Subsection \ref{subsec:2.1}. For each $i\in[n]$, by Lemma \ref{lem:2.1} one can find a function $h_i\in \mL((r+1)\Pin)$ such that $(h_i)=(r+1)(P_i-\Pin)$. Put
\begin{equation}
h_{(i,u)}:=\frac{L_V(\zeta_u h_i)}{h_i}.
\end{equation}
It is clear that $h_{(i,u)}$ is a nonzero function in $\mL((p^l-1)(r+1)\Pin)\subseteq\mL((d-m)\Pin)$.

For $j\in[n]\setminus\{i\}$, we have
\[\Span_{\F_p}\{h_{(i,u)}(P_j):\; u=1,2,\dots,t\}=\Span_{\F_p}\left\{\frac{L_V(\zeta_u h_i(P_j))}{h_i(P_j)}:\; u=1,2,\dots,t\right\}\subseteq \frac{1}{h_i(P_j)}L_V(\F_q).\]
Thus, $b_j^{(i)}\le\dim_{\F_p}(L_V(\F_q))=\log_p q-l$.
The desired result follows from Theorem \ref{thm:3.3}.
\end{proof}

Finally, we consider the Hermitian code $C_L((m-1)\Pin,\mP)$ with $n=r^3$. We label elements of $\mP$ as $\mP=\{P_1,P_2,\dots,P_n\}$. Consider the differential $\eta=\frac{dx}{x^q-x}$. Then $\eta\in\GO((n+2\g-2)\Pin-\sum_{i=1}^nP_i)$. Moreover, $\res_{P_i}(\eta)=1$ for all $i\in[n]$. Since the Euclidean dual of $C_L((n+2\g-2)\Pin, \mP)$ is $C_{\GO}((n+2\g-2)\Pin, \mP)$, we have that for every $g\in\mL((n+2\g-2)\Pin)$
\begin{equation}\sum_{i=1}^ng(P_i)=\sum_{i=1}^n\res_{P_i}(\eta)g(P_i)=0.\end{equation}

\begin{theorem}\label{thm:4.6} If $m\le n+r(r-1)-2-(p^l-1)(r+1)$, then $C_L(m\Pin,\mP)$ is a strong $[n,m, n-1]$-regenerating code with the secondary parameters $(\mathfrak{b}=\log q, M, B)$, where $M=\ell((m-1)\Pin)\log q$ and bandwidth
\begin{equation}B=(n-1)(\log q-l\log p).
\end{equation}
\end{theorem}
\begin{proof} Choose an $\F_p$-subspace $V$ of $\F_q$ of dimension $l$. Consider the $\F_p$-linear map $L_V(x)=\prod_{\Ga\in V}(x-\Ga)$ of $\F_q$ defined in Subsection \ref{subsec:2.1}. For each $i\in[n]$, by Lemma \ref{lem:2.1} one can find a function $h_i\in \mL((r+1)\Pin)$ such that $(h_i)=(r+1)(P_i-\Pin)$. Put
\begin{equation}
h_{(i,u)}:=\frac{L_V(\zeta_u h_i)}{h_i}.
\end{equation}
It is clear that $h_{(i,u)}$ is a nonzero function in $\mL((p^l-1)(r+1)\Pin)\subseteq\mL((n+2\g-2-m)\Pin)$.

Now for any $(f(P_1),f(P_2),\dots,f(P_n))\in C_L(m\Pin,\mP)$, the function $fh_{(i,u)}$ belongs to $\mL((n+2\g-2)\Pin)$. To repair $f(P_i)$, we make use of the following identities
\begin{equation}
\zeta_uf(P_i)=(fh_{(i,u)})(P_i)=-\sum_{1\le j\le n; j\neq i}(fh_{(i,u)})(P_j)=-\sum_{1\le j\le n; j\neq i}h_{(i,u)}(P_j)f(P_j).
\end{equation}

For $j\in[n]\setminus\{i\}$, we have
\[\Span_{\F_p}\{h_{(i,u)}(P_j):\; u=1,2,\dots,t\}=\Span_{\F_p}\left\{\frac{L_V((\zeta_u) h_i(P_j))}{h_i(P_j)}:\; u=1,2,\dots,t\right\}\subseteq \frac{1}{h_i(P_j)}L_V(\F_q).\]
Thus, $b_j^{(i)}\le\dim_{\F_p}(L_V(\F_q))=\log_p q-l$.
The desired result follows from Theorem \ref{thm:3.3}.
\end{proof}

\begin{rmk} Compared with Theorem \ref{thm:4.5}, the bound on code dimension in Theorem \ref{thm:4.6} is less restricted.
\end{rmk}

Finally we give a small example for repairing Hermitian codes and compare them with RS codes.

\begin{ex}{\rm Consider the Hermitian function field with $q=2^6$ and $n=2^9=512$. Then the genus is $\g=28$. Let $m=476$. Then the dimension of the code $C_L((m-1)\Pin,\mP)$ is $m-\g=448$. Thus, the rate of the code is $R=\frac78=1-\frac18$.  Let $p=8$ and $l=1$. We obtain the repairing bandwidth equal $(n-1)(\log q-\log p)=3(n-1)$.

Now we consider a Reed-Solomon code over $\F_{512}$ of length $512$. By  Guruswami-Wootters' result \cite{GW16}, with the same rate $1-\frac 18$, the repairing bandwidth is $(n-1)\log p=3(n-1)$. Thus, both the codes have the same rate and repairing bandwidth. However, the Hermitian code is defined over a smaller alphabet.

}\end{ex}

\end{document}